\documentclass[14pt]{article}
\usepackage{amsmath}
\usepackage{amsthm}
\usepackage{titlesec}
\usepackage{indentfirst}
\usepackage{amssymb}
\usepackage{titlesec}
\usepackage[thinlines,thiklines]{easybmat}

\usepackage{graphicx}
\usepackage{epsfig}
\theoremstyle{definition}\newtheorem{Df}{Definition}
\theoremstyle{plain}\newtheorem{Th}{Theorem}
\theoremstyle{definition}\newtheorem{Rm}{Remark}
\theoremstyle{definition}
\theoremstyle{plain}
\theoremstyle{plain}
\theoremstyle{plain}\newtheorem{Lm}[Th]{Lemma}

\begin{document}
\title{On the state complexity of semi-quantum finite automata}

\author{Shenggen Zheng$^{1,}$\thanks{Corresponding author.
{\it  E-mail
addresss:} zhengshenggen@gmail.com (S. Zheng), gruska@fi.muni.cz (J. Gruska),  issqdw@mail.sysu.edu.cn (D. Qiu).},
\hskip 2mm Jozef Gruska$^{1}$,
\hskip 2mm Daowen Qiu$^{2}$
 \\
\small{{\it $^{1}$ Faculty of Informatics, Masaryk University, Brno 60200, Czech Republic }}\\
\small{{\it $^{2}$ Department of
Computer Science, Sun Yat-sen University, Guangzhou 510006,
  China }}\\
}

\date{ }
\maketitle \vskip 2mm \noindent
{\bf Abstract}\par
Some of the most interesting and important results concerning quantum
finite automata are those showing that they can recognize certain
languages with (much) less resources than corresponding classical finite
automata.  This paper shows three results of such a type that are
stronger in some sense than other ones because (a) they deal with
models of quantum finite automata with very little quantumness (so-called semi-quantum
one- and two-way finite automata);
(b) differences, even comparing with probabilistic classical automata,
are bigger than expected; (c) a trade-off between the number of classical and quantum basis states needed is demonstrated in one case and  (d) languages (or the promise problem) used to show main results are very simple and often explored ones in automata theory or in communication complexity, with seemingly little structure that could be utilized.

\vskip 2mm

\section{Introduction}
An important way to get deeper insights into the power of various quantum resources and operations is to explore the power of various quantum variations of the basic models of classical automata. Of a special interest is to do that for various quantum variations of the classical finite automata, especially for those that use limited amounts of quantum  resources: states, correlations, operations and measurements. This paper aims to contribute to such a line of research.

There are several approaches how to introduce quantum features to
classical models of finite automata. Two of them will be dealt with in this paper. The first one is to consider quantum
variants of the classical {\em one-way (deterministic) finite automata}
(1FA or 1DFA) and the second one is to consider quantum variants of the
classical {\em two-way finite automata} (2FA or 2DFA). Already the very first
attempts to introduce such models, by Moore and Crutchfields \cite{Moo97} as well as
Kondacs and Watrous \cite{Kon97} demonstrated that in spite of the fact that in the
classical case, 1FA and 2FA have the same recognition power, this is not so for their quantum
variations (in case only unitary operations and projective measurements are considered as quantum operations).
Moreover, already the first model of {\em two-way
quantum finite automata} (2QFA), namely that introduced by Kondacs and Watrous,
demonstrated that quantum variants of 2FA are much too
powerful~--~they can recognize even some {\em non-context free languages} and
are actually not really finite  in a strong sense \cite{Kon97}. Therefore it started to be
 of interest to introduce and explore some ``less quantum"
variations of 2FA and their power \cite{Amb02,Amb98,Bro99,Mat12,Yak11}.

A ``hybrid"~--~quantum/clssical~--~variations of 2FA, namely, {\em two-way
finite automata with quantum and classical states}  (2QCFA), were
introduced by Ambainis and Watrous \cite{Amb02}. For this model they showed, in an elegant way, that already an addition of a single qubit to the
classical model can much increase its power. A 2QCFA is
essentially a classical 2FA augmented with a quantum memory of constant
size (for states of a fixed Hilbert space) that does not depend on the
size of the (classical) input. In spite of such a restriction, 2QCFA have
been shown to be even more powerful than {\em two-way probabilistic finite automata}
(2PFA) \cite{Amb02,Zhg12,ZhgQiu11}. A one-way version of 2QCFA was studied in \cite{ZhgQiu112}, namely {\em one-way finite automata with quantum and classical states} (1QCFA).

Number of states is a natural complexity measure for finite automata. In case of quantum finite automata by that we understand the number of the basis states of the quantum  space~--~that is its dimension. In case of hybrid, that is quantum/classical, finite automata, it is natural to consider both complexity measures~--~number of classical and also number of quantum (basis) states~--~and, potentially, trade-offs between them.

 State complexity is one of the important research fields of computer science and it  has many applications \cite{Yu95}, e.g., in natural language and speech processing,
 image generation and encoding, etc.
Early in 1959, Rabin and Scott \cite{Rab59} proved that any $n$-state {\em one-way nondeterministic finite automaton} (1NFA) can be simulated by a $2^n$-state {\em one-way deterministic finite automaton} (1DFA).  Salomaa \cite{Sal64} began to explore state complexity of finite automata in 1960s. The number of states of finite automata used in applications were usually small at that time and therefore  investigations of state complexity of finite automata was seen mainly as a purely theoretical problem. However,  the numbers of states of finite automata in applications can be huge nowadays, even millions of states in some cases \cite{Kir01}.    It becomes therefore also practically  important to explore state complexity of finite automata.
State complexity of several variants of finite automata, both one-way and two-way, were deeply and broadly studied in the past thirty years \cite{AmbNay02,Amb96,Amb02,Amb98,Amb09,AmYa11,Ber05,BMP06,Bir93,Chr86,DwS90,Fre82,Fre08,Fre09,GQZ13,Le06,Liu08,Mer00,Mer01,Mer02,Mil01,Yu94,Yu95,Yak11,Yak10,ZhgQiu112,Zhg12}.

In this paper we  explore the state complexity of semi-quantum finite automata and their space-efficiency comparing to the corresponding classical  counterparts. We do that by showing that even for several very simple, and often considered, languages or promise problems, a little of quantumness can much decrease the state complexity of the corresponding semi-quantum finite automata. The first of these problems will be one of the very basic problem that is explored in communication complexity. Namely, the strings equality problem.

In this paper we  explore the state complexity of semi-quantum finite automata and their space-efficiency comparing to the corresponding classical  counterparts. We do that by showing that even for several very simple, and often considered, languages or promise problems, a little of quantumness can much decrease the state complexity of the corresponding semi-quantum finite automata. The first of these problems will be one of the very basic problem that is explored in communication complexity. Namely, the promise version of  strings equality problem \cite{Buh98,Buh10}.

We use a  promise problem to model the promise version of strings equality problem.
  For the alphabet $\Sigma=\{0,1,\#\}$ and $n\in {\mathbb{Z}}^+$, let us consider the promise problem $A_{EQ}(n)=(A_{yes}(n), A_{no}(n))$, where $A_{yes}(n)=\{x\#y\,|\,x=y,x,y\in\{0,1\}^n\}$ and $A_{no}(n)=\{x\#y\,|\,x\neq y,x,y\in\{0,1\}^n, H(x,y)= \frac{n}{2}\}$. ($H(x,y)$ is the Hamming distance between $x$ and $y$, which is the number of bit positions on which they
differ.)

Klauck \cite{Kla00} has proved that, for any language,  the state complexity of exact  quantum/classical finite automata, which is a general model of one-way  quantum finite automata,   is not less than the state complexity of 1DFA.  Therefore, it is interesting and important to find out whether the result still holds for interesting  cases of promise problems or not\footnote{Ambainis and Yakaryilmaz showed in  \cite{AmYa11} that there is a very special case in which the superiority of quantum
computation to classical one cannot be bounded. }.  Applying the communication complexity result from \cite{Buh98,Buh10} to  finite automata,  for any $n\in {\mathbb{Z}}^+$, we prove  that  promise problem $A_{EQ}(n)$  can be solved  by an exact 1QCFA with $n$ quantum basis states and ${\bf O}({n})$ classical states, whereas the sizes of the corresponding 1DFA are $2^{{\bf \Omega}(n)}$.

As the next we will consider state complexity of the language $L(p)=\{a^{kp}\,|\,k\in {\mathbb{Z}}^+\}$.
It is well know that,  for any $p\in {\mathbb{Z}}^+$, each 1DFA and 1NFA accepting $L(p)$ has at least $p$ states.
Ambainis and Freivalds \cite{Amb98}, proved, using a non-constructive method,  that $L(p)$ can be recognized by a one-way measure-once quantum finite automaton (MO-1QFA) with one-sided error $\varepsilon$  with $poly\left(\frac{1}{\varepsilon} \right)\cdot\log p$ basis states (where $poly(x)$ is some polynomial in $x$). This bound  was improved to ${\bf O}(\frac{\log p}{\varepsilon^3})$ in \cite{BMP06} and  to $4\frac{\log 2p}{\varepsilon}$ in \cite{Amb09}.   That is the best result  known for such a mode of acceptance and it is an interesting open problem whether this bound can be much improved.
If $p$ is a prime, $L(p)$ can not be recognized by any one-way probabilistic finite automaton (1PFA) with less than $p$ states \cite{Amb98}. For the case that $p$ is not a prime, Mereghetti el at. \cite{Mer01} showed that the number of states of a 1PFA necessary and sufficient
for accepting the  language $L(p)$ with isolated cut point  is $p_1^{\alpha_1}+p_2^{\alpha_2}+\cdots+ p_s^{\alpha_s}$, where $p_1^{\alpha_1}p_2^{\alpha_2}\cdots p_s^{\alpha_s}$ is the prime factorization of $p$.  Mereghetti el at. \cite{Mer01} also proved that $L(p)$  can be recognized by a 2 basis states  MO-1QFA with isolated cut point.  However, this mode of acceptance  often leads to  quite different  state complexity outcome than  one-sided error and error probability  acceptance modes.

Concerning  two-way finite automata, for any prime $p$, $p$ states are necessary and sufficient for accepting $L(p)$ on {\em two-way deterministic finite automata} (2DFA) and  {\em two-way nondeterministic finite automata} (2NFA) \cite{Mer00}. For the case that $p$ is  not prime,   the number of states necessary and sufficient for accepting $L(p)$ on 2DFA and 2NFA is $p_1^{\alpha_1}+p_2^{\alpha_2}+\cdots+ p_s^{\alpha_s}$  \cite{Mer00}, where $p_1^{\alpha_1}p_2^{\alpha_2}\cdots p_s^{\alpha_s}$ is the prime factorization of $p$.  Yakaryilmaz and Cem Say \cite{Yak10} showed that there exists a $7$-state {\em one-way finite automaton with restart} (1QFA{$^{\circlearrowleft}$) which accepts $L(p)$ with one-sided error $\varepsilon$ and  expected running time  {\bf O}$(\frac{1}{\varepsilon}\sin^{-2}(\frac{\pi}{p})|w|)$, where $|w|$  is the length of  input $w$. For any $n$-state 1QFA{$^{\circlearrowleft}$ ${\cal M}_1$ with expected running time $t(|w|)$, Yakaryilmaz and Cem Say \cite{Yak10} also proved that there exists a 2QCFA ${\cal M}_2$ with $n$ quantum basis states, ${\bf O}(n)$ classical states, and with expected runtime ${\bf O}(t(|w|))$, such that ${\cal M}_2$
accepts every input string $w$ with the same probability as ${\cal M}_1$ does. Therefore, $L(p)$ can be recognized with one-sided error $\varepsilon$ by a 2QCFA  with $7$ quantum basis states and a constant number of classical states.

In this paper we  prove   that the language $L(p)$ can be recognized  with one-sided error $\varepsilon$ in a linear expected running time ${\bf O}(\frac{1}{ \varepsilon}p^2|w|)$ by a  2QCFA ${\cal A}(p,\varepsilon)$ with 2 quantum basis states and a constant number of classical states. We also show that the number of  states needed for accepting $L(p)$ on a polynomial time 2PFA is at least $\sqrt[3]{(\log p)/b}$, where $b$ is a constant.

The problem of checking
whether the length of  input string is equal to a given constant $m\in {\mathbb{Z}}^+$, is extensively  studied in literatures as well. For any $m\in {\mathbb{Z}}^+$ and any finite alphabet $\Sigma$, it is obvious that the number of states of a 1DFA for accepting the language $C(m)=\{w\,|\,w\in\Sigma^m\}$  is at least  $m$. Freivalds \cite{Fre82} showed that there is an $\varepsilon$ error probability 1PFA accepting $C(m)$ with  ${\bf O}({\log^2 m})$ states.  Ambainis and Freivalds \cite{Amb98} proved that $C(m)$ can be recognized by an MO-1QFA with ${\bf O}(\log m)$ quantum basis states.  Yakaryilmaz and Cem Say \cite{Yak10} showed that there exists a $7$-state 1QFA{$^{\circlearrowleft}$ ${\cal M}$ which accepts $C(m)$ with one-sided error $\varepsilon$ and  expected running time  {\bf O}$(\frac{1}{\varepsilon}2^m|w|)$ which is an exponential of  $m$. The 1QFA{$^{\circlearrowleft}$ ${\cal M}$ can only work efficiently on a very small $m$.

 In this paper we prove that the language $C(m)$ can be recognized with one-sided error  $\varepsilon$ in  expected running time ${\bf O}(\frac{ 1}{\varepsilon}m^2|w|^4)$ by a 2QCFA ${\cal A}(m, \varepsilon)$ with 2 quantum basis states and a constant number of classical states. The expected running time is a polynomial of $m$ and $|w|$. We show also that  the number of  states needed for accepting $C(m)$ on a polynomial 2PFA is at least $\sqrt[3]{(\log m)/b}$, where $b$ is a constant.

 Since 1QCFA and 2QCFA have both quantum and classical states, it is interesting to ask when  there is  some trade-off between these two kinds of states. We prove  such a trade-off property for the case a 1QCFA accepts the  language $L(p)$.  Namely, it holds that for any integer $p$ with prime factorization $p=p_1^{\alpha_1}p_2^{\alpha_2}\cdots p_s^{\alpha_s}$ ($s>1$), for any partition  $I_1, I_2$ of $\{1,\ldots, s\}$,  and for $q_1=\prod_{i\in I_1}p_i^{\alpha_i}$ and $q_2=\prod_{i\in I_2}p_i^{\alpha_i}$, the language $L(p)$  can be recognized with a  one-sided error $\varepsilon$ by a  1QCFA $A(q_1,q_2,\varepsilon)$ with ${\bf O}(\log{q_1})={\bf O}(\sum_{i\in I_1}{\alpha_i}\log p_i)$ quantum basis states and ${\bf O}(q_2)={\bf O}(\prod_{i\in I_2}p_i^{\alpha_i})$ classical states.

The paper is structured as follows. In Section 2 some basic concepts and notations are introduced and automata models involved are
described in some details. State complexities for the string equality problems will be discussed in Section 3. State succinctness for two families of regular languages is explored in Section~4. A trade-off property for 1QCFA is demonstrated in Section 5. Finally, Section 6 contains a conclusion and suggestions for further research.

\section{Preliminaries}

  We introduce in this section some  basic concepts and also notations concerning  quantum  information processing  and afterwards also  the models of 1QCFA and
2QCFA. Concerning more on quantum  information processing we refer the reader to \cite{Gru99,Nie00},
and concerning more on classical and quantum automata \cite{Gru99,Gru00,Hop79,Paz71,Qiu12}.

\subsection{Preliminaries of quantum information processing}

According to  quantum mechanical principles, to each closed quantum system ${\cal S}$, a Hilbert space ${\cal H_S}$ is associated and states of ${\cal S}$ correspond to  vectors of the norm one of ${\cal H_S}$. In case ${\cal H_S}$ is an $n$-dimensional vector space then it has  a basis (actually infinite many of them) consisting of $n$ mutually orthogonal  vectors.
We will mostly denote such a basis and its vectors  by
$$\{|i\rangle\}^{n}_{i=1}.$$
In such a case any vector of ${\cal H_S}$ can be uniquely expressed as a superposition
\begin{equation}
|\psi\rangle=\sum_{i=1}^n\alpha_i|i\rangle,
\end{equation}
where $\alpha_i$'s are complex numbers, called probability amplitudes, that satisfy the following, so-called normalization, condition $\sum_{i=1}^n|\alpha_i|^2=1$. If the state $|\psi\rangle$ is measured with respect to the above basis, then the state collapses to one of the states $|i\rangle$ and to a particular state $|i_0\rangle$ with the probability $|\alpha_{i_0}|^2$.  $i_0$ is then the outcome (discrete) received into the classical world.

Each evolution step of a finite $n$ dimensional quantum system is specified by a
unitary $n\times n$ matrix $U$ and
changes any current state  $|\phi\rangle$ into the state $U|\phi\rangle$.

To extract some information from a quantum state $|\psi\rangle$ a measurement has to be performed.  We will  consider here  mostly measurements
defined by a set $\{P_m\}$ of so-called projective operators/matrices, where indices $m$ refer to the potential classical outcomes of measurements, with the property:
\begin{equation}
P_iP_j=\left\{\begin{array}{ll}
                    P_i\ \ &
i=j,\\
                    0&i\neq j,
                  \end{array}
 \right.
\end{equation}
that, in addition, satisfies the following completeness condition
\begin{equation}\sum_{i=1}^n P_i=I.\end{equation}
In case a state $|\psi\rangle$ is measured with respect to the set of projective operators $\{P_m\}$, then the classical outcome $m$ is obtained with the probability
\begin{equation}p(m)=\|P_m|\psi\rangle\|^2,\end{equation}
and then the state $|\psi\rangle$ ``collapses" into the state
 \begin{equation}\frac{P_m|\psi\rangle}{\sqrt{p(m)}}.\end{equation}

A projective measurement $\{P_m\}$ is usually specified by an {\em observable} $M$, a Hermitian matrix
that has a so called spectral decomposition
\begin{equation}M=\sum_m mP_m,\end{equation}
where $m$ are mutually different eigenvalues of $M$ and each $P_m$ is a projector into the space of eigenvectors associated to the eigenvalue $m$.

\subsection{Preliminaries on semi-quantum finite automata}
2QCFA were introduced by Ambainis and Watrous \cite{Amb02} and explored also by Yakaryilmaz,  Qiu, Zheng and others \cite{Qiu08,Yak10,ZhgQiu112,Zhg12,ZhgQiu11,Zhg13a}. Informally, a 2QCFA can be seen as a 2DFA with an access to a quantum memory for states of a fixed Hilbert space upon which at each step either a unitary operation is performed or a projective measurement and the outcomes of which then probabilistically determine the next move of the underlying 2DFA.

 \begin{Df}
A 2QCFA ${\cal A}$ is specified by a 9-tuple
\begin{equation}
{\cal A}=(Q,S,\Sigma,\Theta,\delta,|q_{0}\rangle,s_{0},S_{acc},S_{rej})
\end{equation}
where:

\begin{enumerate}
\item $Q$ is a finite set of orthonormal quantum basis states.
\item $S$ is a finite set of classical states.
\item $\Sigma$ is a finite alphabet of input symbols and let
$\Sigma'=\Sigma\cup \{|\hspace{-1.5mm}c,\$\}$, where $|\hspace{-1.5mm}c$ will be used as the left end-marker and $\$$ as the right end-marker.
\item $|q_0\rangle\in Q$ is the initial quantum state.
\item $s_0$ is the initial classical state.
\item $S_{acc}\subset S$ and $S_{rej}\subset S$, where $S_{acc}\cap S_{rej}=\emptyset$ are  sets of
the classical accepting and rejecting states, respectively.
\item $\Theta$ is a quantum transition function
\begin{equation}
\Theta: S\setminus(S_{acc}\cup S_{rej})\times \Sigma'\to U(H(Q))\cup O(H(Q)),
\end{equation}
where U(H(Q)) and O(H(Q)) are sets of unitary operations and projective measurements on the Hilbert space generated by quantum states from $Q$.

\item $\delta$ is a classical transition function.
If the automaton ${\cal A}$ is in the classical state $s$, its tape head is  scanning a symbol $\sigma$ and its quantum memory is in the quantum state $|\psi\rangle$, then ${\cal A}$ performs quantum and classical transitions  as follows.
\begin{enumerate}
\item If $\Theta(s,\sigma)\in U(H(Q))$, then the unitary operation $\Theta(s,\sigma)$ is applied on the current state $|\psi\rangle$ of quantum memory to produce a new quantum state. The automaton performs, in addition, the following classical transition function
\begin{equation}
\delta: S\setminus(S_{acc}\cup S_{rej})\times \Sigma'\to S\times \{-1, 0,1\}.
\end{equation}
If $\delta(s,\sigma)=(s',d)$, then the new classical state of the automaton is $s'$ and its head moves in the direction $d$.

\item If $\Theta(s,\sigma)\in O(H(Q))$, then the measurement operation $\Theta(s,\sigma)$ is applied on the current state $|\psi\rangle$.
 Suppose the  measurement $\Theta(s,\sigma)$ is specified by operators $\{P_1,\ldots, P_n\}$  and its corresponding classical outcome is from the set $N_{\Theta(s,\sigma)}=\{1,2,\cdots,n\}$.
The classical transition function $\delta$ can be then specified as follow
\begin{equation}
\delta: S\setminus(S_{acc}\cup S_{rej})\times \Sigma'\times N_{\Theta(s,\sigma)}\to S\times \{-1, 0,1\}.
\end{equation}
In such a case,  if $i$ is the classical outcome of the measurement, then the
current quantum state $|\psi\rangle$ is changed to
the  state $P_{i}|\psi\rangle/ \|P_{i}|\psi\rangle\|$. Moreover,  if
$\delta{(s,\sigma)}(i) =(s',d)$, then  the new classical state of the automaton is $s'$ and its head moves in the direction $d$.
\end{enumerate}
The automaton halts and accepts (rejects) the input when it enters a classical accepting (rejecting) state (from $S_{acc}$($S_{rej}$)).

\end{enumerate}
\end{Df}

The computation of a 2QCFA
${\cal A}=(Q,S,\Sigma,\Theta,\delta,|q_{0}\rangle,s_{0},S_{acc},S_{rej})$ on an input $w\in \Sigma^*$ starts with the string $|\hspace{-1.5mm}cx\$$ on the input tape. At the start, the tape head of the automation is positioned on the left end-marker and the automaton begins the computation in the classical initial state and
in the initial quantum state. After that,
in each  step, if  its  classical state  is $s$, its tape head reads a symbol $\sigma$ and its quantum state is $|\psi\rangle$, then the automaton changes its states and makes its head movement following the steps described in the definition.

 The computation will end whenever the resulting classical state
is in \linebreak[0] {$S_{acc}\cup S_{rej}$}. Therefore, similarly to the definition
of accepting and rejecting probabilities  for 2QFA \cite{Kon97},
the accepting and rejecting probabilities $Pr[{\cal A}\ {accepts}\  w]$ and $Pr[{\cal A}\ \text{rejects}\  w]$ for an input $w$ are, respectively, the sums of all
accepting probabilities and all rejecting probabilities before the
end of  computation on the input $w$.

\begin{Rm}
1QCFA are one-way versions of 2QCFA \cite{ZhgQiu112}. In this paper, we only use 1QCFA in which a unitary transformation is applied in every step after scanning a symbol and an measurement is performed  after scanning the right end-marker. Such model is an measure-once 1QCFA and corresponds to a variant of MO-1QFA.
\end{Rm}

Three basic modes of language acceptance to be considered here are the following ones:
Let $L\subset \Sigma^*$ and $0<\varepsilon\leq\frac{1}{2}$. A finite automaton ${\cal A}$ recognizes $L$ with a {\em one-sided error} $\varepsilon$ if, for $w\in \Sigma^*$,
\begin{enumerate}
\item[1.] $\forall w\in L$, $Pr[{\cal A}\  \text{accepts}\  w]=1$, and
\item[2.] $\forall w\notin L$, $Pr[{\cal A}\ \text{rejects}\  w]\geq 1-\varepsilon$.
\end{enumerate}

Let $0<\varepsilon<\frac{1}{2}$. A finite automaton ${\cal A}$ recognizes $L$ with an {\em error probability} $\varepsilon$ if, for $w\in \Sigma^*$,
\begin{enumerate}
\item[1.] $\forall w\in L$, $Pr[{\cal A}\  \text{accepts}\  w]\geq 1-\varepsilon$, and
\item[2.] $\forall w\notin L$, $Pr[{\cal A}\ \text{rejects}\  w]\geq 1-\varepsilon$.
\end{enumerate}

Let  $0<\lambda<1$.
A language $L$ is said to be accepted by a finite automation   ${\cal A}$  with an {\em isolated cut point} $\lambda$ if there exists $\delta>0$, for $w\in \Sigma^*$, such that
\begin{enumerate}
\item[1.] $\forall w\in L$, $Pr[{\cal A}\  \text{accepts}\  w]\geq \lambda+\delta$, and
\item[2.] $\forall w\notin L$, $Pr[{\cal A}\ \text{accepts}\  w]\leq \lambda-\delta$.
\end{enumerate}

Obviously, for $0<\varepsilon<\frac{1}{2}$,
 one-sided error acceptance is stricter than an error probability acceptance and the error probability acceptance is stricter than an isolated cut point acceptance.

Language acceptance is a special case of so called promise problem solving.
A {\em promise problem} is a pair $A = (A_{yes}, A_{no})$, where $A_{yes}$, $A_{no}\subset \Sigma^*$
are disjoint sets. Languages may be viewed as promise problems that obey the additional constraint
$A_{yes}\cup A_{no}=\Sigma^*$.

 A promise problem $A = (A_{yes}, A_{no})$ is solved by an exact 1QCFA ${\cal A}$ if
\begin{enumerate}
\item[1.] $\forall w\in A_{yes}$, $Pr[{\cal A}\  \text{accepts}\  w]=1$, and
\item[2.] $\forall w\in  A_{no}$, $Pr[{\cal A}\ \text{rejects}\  w]=1$.
\end{enumerate}

\section{State complexities for strings equality problems}

Strings equality problem is a basic problem in communication complexity \cite{KusNis97b} defined as follows.  Let Alice and Bob
be the two communicating parties. Alice is given as an input $x\in\{0,1\}^n$ and Bob is given as an input $y\in\{0,1\}^n$. They wish to compute the value of the function $EQ(x,y)$  defined to be 1 if $x=y$ and 0 otherwise.

We use a  promise problem to model the promise version of strings equality problem studied in \cite{Buh98,Buh10}.
  For the alphabet $\Sigma=\{0,1,\#\}$ and $n\in {\mathbb{Z}}^+$, let us consider the promise problem $A_{EQ}(n)=(A_{yes}(n), A_{no}(n))$, where $A_{yes}(n)=\{x\#y\,|\,x=y,x,y\in\{0,1\}^n\}$ and $A_{no}(n)=\{x\#y\,|\,x\neq y,x,y\in\{0,1\}^n, H(x,y)= \frac{n}{2}\}$.

\begin{Th}\label{A_{EQ}(n)}
The promise problem $A_{EQ}(n)$ can be solved  by an exact 1QCFA  ${\cal A}(n)$ with $n$ quantum basis states and ${\bf O}(n)$ classical states, whereas the sizes of the corresponding 1DFA  are $2^{{\bf \Omega}(n)}$.
\end{Th}

\begin{figure}[htbp]
 %  %Requires \usepackage{graphicx}
\begin{tabular}{|l|}
    \hline

\begin{minipage}[t]{0.93\textwidth}
\begin{enumerate}
\item[1.] Read the left end-marker $\ |\hspace{-1.5mm}c$,  perform $U_s$ on the initial quantum state $|1\rangle$,  change its classical state to $\delta(s_0,\ |\hspace{-1.5mm}c )=s_1$, and move the tape head one cell to the right.

\item[2.] Until the currently  scanned symbol $\sigma$ is not $\#$, do the following:
 \begin{enumerate}
 \item[2.1] Apply $\Theta(s_i,\sigma)=U_{i,\sigma}$ to the current quantum state.
 \item[2.2] Change the classical state $s_i$ to $s_{i+1}$ and move the tape head one cell to the right.
\end{enumerate}
\item[3.] Change the classical state $s_{n+1}$ to  $s_1$ and move the tape head one cell to the right.

\item[4.] While the currently  scanned symbol $\sigma$ is not the right end-marker $\$$, do the following:
 \begin{enumerate}
 \item[2.1] Apply $\Theta(s_i,\sigma)=U_{i,\sigma}$ to the current quantum state.
 \item[2.2] Change the classical state $s_i$ to $s_{i+1}$ and move the tape head one cell to the right.
\end{enumerate}

\item[5.] When the right end-marker  is reached,    perform $U_{f}$ on the current quantum state,
measure the current quantum state with $M=\{P_i=|i\rangle\langle i|\}_{i=1}^{n}$.   If the outcome is $|1\rangle$, accept the input; otherwise reject the input.

\end{enumerate}

\end{minipage}\\

\hline
\end{tabular}
 \centering\caption{  Description of the behavior of ${\cal A}(n)$ when solving the promise problem $A_{EQ}(n)$. }\label{f3}
\end{figure}
\begin{proof}
Let $x=x_1\cdots x_n$ and $y=y_1\cdots y_n$ with $x,y\in\{0,1\}^n$.  Let us consider a 1QCFA ${\cal A}(n)$ with $n$   quantum basis states
$\{|i\rangle:i=1,2,\ldots,n\}$. ${\cal A}(n)$ will start in  the
quantum state $|1\rangle=(1,0,\ldots,0)^T$. We use classical states $s_i\in S$ ($1\leq i\leq n+1$) to point out the positions of the tape head that will provide some information for  quantum transformations. If  the classical state of ${\cal A}(n)$ will be $s_i$ ($1\leq i\leq n$) that will  mean that the next scanned symbol of the tape head is the $i$-th symbol of $x$($y$) and $s_{n+1}$ means that the next scanned symbol of  the tape head  is  $\#$($\$$).
 The automaton proceeds as shown in Figure \ref{f3}, where
\begin{eqnarray*}
&U_s|1\rangle=\frac{1}{\sqrt{n}}\sum_{i=1}^n|i\rangle;\\
&U_{i,\sigma}|i\rangle=(-1)^{\sigma}|i\rangle \ \  \mbox{and } \ \  U_{i,\sigma}|j\rangle=|j\rangle \ \mbox{for}\ j\neq i;\\
&U_f(\sum_{i=1}^n\alpha_i|i\rangle)=(\frac{1}{\sqrt{n}}\sum_{i=1}^n\alpha_i)|1\rangle+\cdots.
\end{eqnarray*}

 Transformations $U_{s}$ and $U_{f}$ are unitary. The first column of $U_{s}$ is $\frac{1}{\sqrt{n}}(1,\ldots,1)^T$ and the first row of $U_f$ is $\frac{1}{\sqrt{n}}(1,\ldots,1)$.

The quantum state after scanning the left end-marker is $|\psi_1\rangle= U_s|1\rangle=\sum_{i=1}^n\frac{1}{\sqrt{n}}\linebreak[0]|i\rangle$, the quantum state after Step 2 is $|\psi_2\rangle=\sum_{i=1}^n\frac{1}{\sqrt{n}}(-1)^{x_i}|i\rangle$, and the quantum state after Step 4 is $|\psi_{3}\rangle=\sum_{i=1}^n\frac{1}{\sqrt{n}}(-1)^{x_i+y_i}|i\rangle$. The quantum state after scanning  the right end-marker  is therefore
\begin{equation}
|\psi_4\rangle=U_f\left(\sum_{i=1}^n\frac{1}{\sqrt{n}}(-1)^{x_i+y_i}|i\rangle\right)=U_f\frac{1}{\sqrt{n}} \left(
                                                  \begin{array}{c}
                                                    (-1)^{x_1+y_1} \\
                                                    (-1)^{x_2+y_2} \\
                                                    \vdots \\
                                                    (-1)^{x_n+y_n}\\
                                                  \end{array}
                                                \right)
                                                \end{equation}
\begin{equation}
 =\left(
                                                  \begin{array}{c}
                                                    \frac{1}{n} \sum_{i=1}^n (-1)^{x_i+y_i} \\
                                                    \vdots \\
                                                    \vdots \\
                                                  \end{array}
                                                \right).
  \end{equation}
 % \sum_{i=1}^n(-1)^{x_i+y_i}\sum_{j=1}^n(-1)^{i\cdot j}|j\rangle.

If the input string $w\in A_{yes}(n)$, then $x_i=y_i$ for $1\leq i\leq n$ and  $|\frac{1}{n} \sum_{i=1}^n (-1)^{x_i+y_i}|^2\linebreak[0]=1$. The amplitude of $|1\rangle$ is 1, and that means $|\psi_4\rangle=|1\rangle$.
Therefore the input will be accepted with probability 1 at the measurement in Step 5.

If the input string $w\in A_{no}(n)$, then $H(x,y)=\frac{n}{2}$. Therefore   the probability of getting outcome $|1\rangle$ in the measurement in Step 5 is $|\frac{1}{n} \sum_{i=1}^n (-1)^{x_i+y_i}|^2=0$.

The deterministic communication complexity for the promise version of  strings equality problem is at least $0.007n$ \cite{Buh98,Buh10}. Therefore,  the sizes of the corresponding 1DFA are $2^{{\bf \Omega}(n)}$ \cite{KusNis97b}.
\end{proof}

\section{State succinctness for 2QCFA}

 State succinctness for 2QCFA was  explored by  Yakaryilmaz, Zheng and others \cite{Yak10,Zhg12}. In \cite{Zhg12},
 Zheng et al. showed the state succinctness for polynomial time 2QCFA for  families of promise problems and for exponential time 2QCFA for a family of languages. In this section, we  show the state succinctness for linear time 2QCFA and polynomial time 2QCFA for two families of languages.

 \subsection{State succinctness for the language $L(p)$}

\begin{Th}
For any $p\in {\mathbb{Z}}^+$ and $0<\varepsilon\leq\frac{1}{2}$, the language $L(p)$  can be recognized  with one-sided error $\varepsilon$  by  a 2QCFA ${\cal A}(p,\varepsilon)$ with 2 quantum basis states and a  constant number of classical states (neither depending on $p$ nor on $\varepsilon$) in  a  linear expected running time ${\bf O}(\frac{1}{ \varepsilon}p^2n)$,   where $n$ is the length of  input.
\end{Th}

\begin{proof}
The main idea of the proof is as follows: We consider a 2QCFA ${\cal A}(p,\varepsilon)$ with 2 orthogonal  quantum basis states
$|q_0\rangle$ and $|q_1\rangle$. ${\cal A}(p,\varepsilon)$ starts computation in the initial
quantum state $|q_0\rangle$ and with the tape head on the left end-marker. Every time when ${\cal A}(p,\varepsilon)$ reads a symbol `a', the current quantum state is rotated by the angle $\frac{\pi}{p}$. When the right end-marker $\$$ is reached, ${\cal A}(p,\varepsilon)$ measures the current quantum state. If the resulting quantum state is $|q_1\rangle$, the input string  is rejected, otherwise the automaton proceeds as shown in Figure \ref{f1}, where

\begin{figure}[htbp]
 %  %Requires \usepackage{graphicx}
\begin{tabular}{|l|}

    \hline
    \begin{minipage}[t]{0.93\textwidth}
Repeat the following  ad infinity:
 \begin{enumerate}
\item[1.] Move the tape head to the right of the left end-marker.
\item[2.] Until the scanned symbol is the right end-marker, apply $U_{p}$ to the current
 quantum state and move the head one cell to the right.
\item[3] Measure the current quantum state in the basis $\{|q_0\rangle, |q_1\rangle\}$.
 \begin{enumerate}
\item[3.1] If quantum outcome is $|q_1\rangle$, reject the input.
\item[3.2] Otherwise apply $U_{p,\varepsilon}$ to the current quantum state $|q_0\rangle$.
 \end{enumerate}
\item[4] Measure the quantum state in the basis $\{|q_0\rangle, |q_1\rangle\}$. If the result is $|q_0\rangle$,
 accept the input; otherwise apply a unitary operation to change the quantum state from $|q_1\rangle$ to $|q_0\rangle$
 and start a new iteration.
\end{enumerate}
   \end{minipage}\\
\hline
\end{tabular}
 \centering\caption{ Description of the behavior of ${\cal A}(p,\varepsilon)$ when recognizing the language $L({p})$. }\label{f1}
\end{figure}

\begin{equation}\label{matrix}
U_{p}=\left(
  \begin{array}{cc}
    \cos \frac{\pi}{p}  & -\sin \frac{\pi}{p} \\
    \sin \frac{\pi}{p}  & \cos \frac{\pi}{p} \\
  \end{array}
\right) \ \ {and}\ \
 U_{p,\varepsilon}=\left(
  \begin{array}{cc}
    \frac{1}{\sqrt{p^2/4\varepsilon}}  & -\frac{\sqrt{p^2/4\varepsilon-1}}{\sqrt{p^2/4\varepsilon}}\\
    \frac{\sqrt{p^2/4\varepsilon-1}}{\sqrt{p^2/4\varepsilon}}  & \frac{1}{\sqrt{p^2/4\varepsilon}}  \\
  \end{array}
\right).
\end{equation}
\begin{Lm}\label{in-L(p)}
If the input $w\in L(p)$, then the quantum state of  ${\cal A}(p,\varepsilon)$ after Step 2
is one of the quantum states $\pm|q_0\rangle$.
\end{Lm}
\begin{proof}
If $w\in L(p)$, then $|w|=n=kp$, where $k\in {\mathbb{Z}}^+$.
Starting with the state $|q_0\rangle$, ${\cal A}(p,\varepsilon)$ changes its quantum state to  $|q\rangle=(U_{p})^n|q_0\rangle$ after Step 2, where
\begin{equation}
|q\rangle=(U_{p})^n|q_0\rangle=\left(
  \begin{array}{cc}
    \cos \frac{\pi}{p}  & -\sin \frac{\pi}{p} \\
    \sin \frac{\pi}{p}  & \cos \frac{\pi}{p} \\
  \end{array}
\right)^n|q_0\rangle=\left(
  \begin{array}{cc}
    \cos \frac{n\pi}{p}  & -\sin \frac{n\pi}{p} \\
    \sin \frac{n\pi}{p}  & \cos \frac{n\pi}{p} \\
  \end{array}
\right)|q_0\rangle
\end{equation}
\begin{equation}
=\left(
  \begin{array}{cc}
    \cos k\pi  & -\sin  k\pi \\
    \sin  k\pi  & \cos  k\pi \\
  \end{array}
\right)|q_0\rangle=\pm|q_0\rangle.
\end{equation}
\end{proof}
\begin{Lm}\label{not-in-L-p}
If the input $w\not\in L(p)$, then ${\cal A}(p,\varepsilon)$ rejects $w$ after Step 3 with a probability at least $4/p^2$.
\end{Lm}
\begin{proof}
Suppose $n=|w|=kp+i$, where $k\in {\mathbb{Z}}^+$ and $i\in \{1,2,\cdots,p-1\}$. The quantum state of ${\cal A}(p,\varepsilon)$ after Step 2 will be
\begin{equation}
|q\rangle=(U_{p})^n|q_0\rangle=\left(
  \begin{array}{cc}
    \cos \frac{\pi}{p}  & -\sin \frac{\pi}{p} \\
    \sin \frac{\pi}{p}  & \cos \frac{\pi}{p} \\
  \end{array}
\right)^n|q_0\rangle=\left(
  \begin{array}{cc}
    \cos \frac{n\pi}{p}  & -\sin \frac{n\pi}{p} \\
    \sin \frac{n\pi}{p}  & \cos \frac{n\pi}{p} \\
  \end{array}
\right)|q_0\rangle
\end{equation}
\begin{equation}
=\left(
  \begin{array}{cc}
    \cos \frac{i\pi}{p}  & -\sin \frac{i\pi}{p}  \\
    \sin \frac{i\pi}{p}  & \cos \frac{i\pi}{p}  \\
  \end{array}
\right)|q_0\rangle=\cos \frac{i\pi}{p}  |q_0\rangle+\sin \frac{i\pi}{p}  |q_1\rangle.
\end{equation}
The probability of observing  $|q_1\rangle$ is
$\sin^2 \frac{i\pi}{p}$ in Step 3.

Let $f(x)=sin(x\pi)-2x$. We have
 $f''(x)=-\pi^2\sin(x\pi)\leq 0$ when $x\in
[0,1/2]$. Therefore, $f(x)$ is concave in  the interval $[0,1/2]$, and $f(0)=f(1/2)=0$. So, for any $x\in[0,1/2]$,
$f(x)\geq 0$, that is $\sin(x\pi)\geq 2x$.  Therefore,
\begin{equation}
\sin^2 \frac{i\pi}{p}\geq \sin^2 \frac{\pi}{p}\geq \left(\frac{2}{p}\right)^2=\frac{4}{p^2}.
\end{equation}
\ \
\end{proof}

\begin{Lm}\label{acc-rej}
If the input $w\in L(p)$, then ${\cal A}(p,\varepsilon)$ accepts $w$ after Step 4 with the probability $\frac{4 \varepsilon}{p^2}$.  If the input $w\not\in L(p)$, then ${\cal A}(p,\varepsilon)$ accepts $w$ after Step 4 with a probability less than $\frac{4 \varepsilon}{p^2}$.
\end{Lm}
\begin{proof}
If $w\in L(p)$, then the quantum state of ${\cal A}(p,\varepsilon)$ after Step 2 will be $|q_0\rangle$, according to Lemma \ref{in-L(p)}. After Step 3 the quantum state will be
\begin{equation}
|q\rangle=U_{p,\varepsilon}|q_0\rangle=\left(
  \begin{array}{cc}
    \frac{1}{\sqrt{p^2/4\varepsilon}}  & -\frac{\sqrt{p^2/4\varepsilon-1}}{\sqrt{p^2/4\varepsilon}}\\
    \frac{\sqrt{p^2/4\varepsilon-1}}{\sqrt{p^2/4\varepsilon}}  & \frac{1}{\sqrt{p^2/4\varepsilon}}  \\
  \end{array}\right)|q_0\rangle= \frac{1}{\sqrt{p^2/4\varepsilon}}|q_0\rangle+ \frac{\sqrt{p^2/4\varepsilon-1}}{\sqrt{p^2/4\varepsilon}}|q_1\rangle.
\end{equation}
Therefore, after the measurement in Step 4, the input is accepted with the probability $\frac{4 \varepsilon}{p^2}$.

 If $w\not\in L(p)$, then the probability that the quantum state of ${\cal A}(p,\varepsilon)$ after Step 2 is $|q_0\rangle$ is less than 1. Therefore, ${\cal A}(p,\varepsilon)$ accepts $w$ after Step 4 with a probability less than $\frac{4 \varepsilon}{p^2}$.
\end{proof}

In Step 4, if ${\cal A}(p,\varepsilon)$ accepts
the input string, then ${\cal A}(p,\varepsilon)$ halts.  Otherwise, the
resulting quantum state of the measurement is $|q_1\rangle$ and an application of the operation $|q_0\rangle\langle q_1|$ results in state $|q_0\rangle$ in which then
${\cal A}(p,\varepsilon)$ starts a new iteration.

If $w\in L(p)$, then according to Lemma \ref{acc-rej}, the probability of accepting the input in one iteration is
\begin{equation}
 P_a=\frac{4 \varepsilon}{p^2}
\end{equation}
and the probability of rejecting the input in one iteration is
\begin{equation}
 P_r=0.
\end{equation}
If  the whole process is repeated for ad infinitum, then the accepting probability is
\begin{equation}
Pr[{\cal A}(p,\varepsilon)\  \text{accepts}\  w] =\sum_{i\geq
0}(1-P_a)^i(1-P_r)^iP_a
=\sum_{i\geq
0}(1-P_a)^iP_a
=\frac{P_a}{P_a}=1.
\end{equation}

If $w\not\in L(p)$, then according to Lemmas  \ref{not-in-L-p}  and \ref{acc-rej}, the probability of accepting the input in one iteration is
\begin{equation}
 P_a<\frac{4 \varepsilon}{p^2}
\end{equation}
and the probability of rejecting the input in one iteration is
\begin{equation}
 P_r>\frac{4}{p^2} .
\end{equation}
If the whole process is repeated indefinitely, then the probability that ${\cal A}(p,\varepsilon)$ rejects the input $w$  is
\begin{equation}
Pr[{\cal A}(p,\varepsilon)\  \text{rejects}\  w] =\sum_{i\geq
0}(1-P_a)^i(1-P_r)^iP_r
=\frac{P_r}{P_a+P_r-P_aP_r}>\frac{P_r}{P_a+P_r}
\end{equation}
\begin{equation}
>\frac{\frac{4}{p^2}}{\frac{4 \varepsilon}{p^2}+\frac{4}{p^2}}
=\frac{1}{\varepsilon+1}>1-\varepsilon.
\end{equation}
{\bf Time analysis:} Steps 1 to 4 take  time ${\bf O}(n)$. The halting probability  is in both cases  ${\bf \Omega}\left(\frac{ \varepsilon}{p^2}\right)$, so the expected number of repetitions of the above process is, in both cases, {\bf O}($\frac{p^2}{ \varepsilon}$).
Hence the expected running time of ${\cal A}(p,\varepsilon)$ is
{\bf O}($\frac{1}{ \varepsilon}p^2n$).
\end{proof}

The most important work in designing an automaton is to design its transition function.
Similar unitary matrixes of $U_p$ (the quantum transition function after scanning a symbol of input) and proof methods in the previous Theorem can be found in \cite{Amb98,Amb09,Mer01,BMP06,Yak10}. The state complexity for $L(P)$ is relative to the error probability ${\varepsilon}$ in
\cite{Amb98,Amb09,Mer01,BMP06}.
  But in this paper  the state complexity for $L(P)$ is not relative to the error probability ${\varepsilon}$, since we use a
special matrix $ U_{p,\varepsilon}$, which is never used in other papers.

\begin{Th}\label{th-EQtoDFA}
For any integer $p$, any  polynomial expected running time 2PFA recognizing $L(p)$ with error probability $\varepsilon<\frac{1}{2}$ has at least $\sqrt[3]{(\log p)/b}$ states, where $b$ is a constant.
\end{Th}

In order to prove this theorem, we need

\begin{Lm}[\cite{DwS90}]\label{2PFAtoDFA}
For every $\varepsilon<1/2$, $a>0$ and $d>0$, there exists a constant $b>0$ such that, for any integer $c$, if a language $L$ is recognized with an error probability $\varepsilon$  by a $c$-state 2PFA within time $an^d$, where $n=|w|$ is the length of  input, then $L$ is recognized by some DFA with at most $c^{bc^2}$ states.
\end{Lm}

\begin{proof}Assume that a $c$-state 2PFA ${\cal A}(p)$ recognizes
$L(p)$ with an error probability $\varepsilon<1/2$ and also
within a polynomial expected running time. According to Lemma
\ref{2PFAtoDFA}, there exits a 1DFA that recognizes $L(p)$ with
$c^{bc^2}$ states, where $b>0$ is a constant. As we know, any DFA recognizing $L(p)$ has at least $p$ states. Therefore,
\begin{equation}
c^{bc^2}\geq p \Rightarrow b c^2\log{c}\geq \log p
\Rightarrow c^3> (\log p)/ b\Rightarrow c>\sqrt[3]{(\log p)/b}.
\end{equation}
\ \
\end{proof}

\subsection{State succinctness for the language $C(m)$}

\begin{Th}\label{L-m}
For any $m\in {\mathbb{Z}}^+$ and $0<\varepsilon\leq\frac{1}{2}$, the language $C(m)$  can be recognized with one-sided error $\varepsilon$ by  a  2QCFA ${\cal A}(m,\varepsilon)$ with 2 quantum basis states and  a constant number of classical states (neither depending on $m$ nor on $\varepsilon$)  in a polynomial expected running time ${\bf O}(\frac{ 1}{\varepsilon}m^2n^4)$,  where $n$ is the length of  input.
\end{Th}
\begin{figure}[htbp]
 %  %Requires \usepackage{graphicx}
\begin{tabular}{|l|}
    \hline
\begin{minipage}[t]{0.93\textwidth}
Repeat the following  ad infinity:
 \begin{enumerate}
\item[1.] Move the tape head to the left end-marker, read the  end-marker $\ |\hspace{-1.5mm}c$,  apply
 $U_{|\hspace{-1mm}c}$ on $|q_0\rangle$, and move the tape head one cell to the right.
\item[2.] Until the scanned symbol is the right end-marker, apply $U_{\alpha}$ to the current
 quantum state and move the tape head one cell to the right.

\item[3.0] When the right end-marker is reached, measure the quantum state in the
 basis $\{|q_0\rangle, |q_1\rangle\}$.
 \begin{enumerate}
\item[3.1] If quantum outcome is $|q_1\rangle$, reject the input.
\item[3.2] Otherwise repeat the following subroutine two times:
 \begin{enumerate}
\item[3.2.1] Move the tape head to the first symbol right to the left end-marker.
\item[3.2.2] Until the currently read symbol is one of the end-markers simulate  a
 coin-flip and move the head right (left) if the outcome of the coin-flip is
 ``head" (``tail").
 \end{enumerate}
  \end{enumerate}

\item[4.] If the above process ends both times at the right end-marker,  apply $U_{m,\varepsilon}$
  to the current quantum state and measure the quantum state in the basis
  $\{|q_0\rangle, |q_1\rangle\}$. If the result is $|q_0\rangle$, accept the input;  otherwise apply a unitary operation to change the quantum state from $|q_1\rangle$ to $|q_0\rangle$
 and start a new iteration.
 \end{enumerate}
\end{minipage}\\
\hline
\end{tabular}
 \centering\caption{ Description of the behavior of ${\cal A}(m,\varepsilon)$ when recognizing the language $L_{m}$. }\label{f2}
\end{figure}

\begin{proof}
The main idea of the proof is as follows: we consider a 2QCFA ${\cal A}(m,\varepsilon)$ with 2  orthogonal quantum basis states
$|q_0\rangle$ and $|q_1\rangle$. ${\cal A}(m,\varepsilon)$ starts computation with the initial
quantum state $|q_0\rangle$. When ${\cal A}(m,\varepsilon)$ reads the left end-marker$\ |\hspace{-1.5mm}c$, the current quantum state will be rotated by the angle $-\sqrt{2}m\pi$ and every time when ${\cal A}(m,\varepsilon)$ reads a new symbol $\sigma\in\Sigma$, the state is rotated by the angle $\alpha=\sqrt 2\pi$ (notice that $\sqrt{2}m\pi=m\alpha$). When the right end-marker $\$$ is reached, ${\cal A}(m,\varepsilon)$ measures the current quantum state with projectors $\{|q_0\rangle\langle q_0|,|q_1\rangle\langle q_1|\}$. If the resulting quantum state is $|q_1\rangle$, the input string $w$ is rejected, otherwise, the automaton proceeds as shown in Figure \ref{f2}, where
\begin{equation}\label{matrix}
U_{|\hspace{-1.1mm}c}=\left(
  \begin{array}{cc}
    \cos m\sqrt{2}\pi  & \sin m\sqrt{2}\pi\\
    -\sin m\sqrt{2}\pi  & \cos  m\sqrt{2}\pi\\
  \end{array}
\right),\
U_{\alpha}=\left(
  \begin{array}{cc}
    \cos \sqrt{2}\pi & -\sin  \sqrt{2}\pi  \\
    \sin  \sqrt{2}\pi   & \cos  \sqrt{2}\pi  \\
  \end{array}
\right),
 \end{equation}
  \begin{equation}
 U_{m,\varepsilon}=\left(
  \begin{array}{cc}
    \frac{1}{\sqrt{2m^2/\varepsilon}}  & -\frac{\sqrt{2m^2/\varepsilon-1}}{\sqrt{2m^2/\varepsilon}}\\
    \frac{\sqrt{2m^2/\varepsilon-1}}{\sqrt{2m^2/\varepsilon}}  & \frac{1}{\sqrt{2m^2/\varepsilon}}  \\
  \end{array}\right).
\end{equation}

\begin{Lm}\label{in-C(m)}
If the input $w\in C(m)$, then the quantum state of the above automaton ${\cal A}(m,\varepsilon)$, after Step 2,
is  $|q_0\rangle$.
\end{Lm}
\begin{proof}
For $w\in C(m)$, we have $|w|=m$.
Starting with the state $|q_0\rangle$, ${\cal A}(m,\varepsilon)$ changes its quantum state after processing the whole input after Step 2 to
\begin{align}
|q\rangle&=(U_{\alpha})^mU_{|\hspace{-1.1mm}c}|q_0\rangle=
\left(
  \begin{array}{cc}
    \cos \sqrt{2}\pi & -\sin  \sqrt{2}\pi  \\
    \sin  \sqrt{2}\pi   & \cos  \sqrt{2}\pi  \\
  \end{array}
\right)^m
\left(
  \begin{array}{cc}
    \cos m\sqrt{2}\pi  & \sin m\sqrt{2}\pi \\
    -\sin m\sqrt{2}\pi  & \cos  m\sqrt{2}\pi \\
  \end{array}
\right)|q_0\rangle\\
&=\left(
  \begin{array}{cc}
    \cos m\sqrt{2}\pi & -\sin  m\sqrt{2}\pi  \\
    \sin  m\sqrt{2}\pi   & \cos  m\sqrt{2}\pi  \\
  \end{array}
\right)
\left(
  \begin{array}{cc}
    \cos m\sqrt{2}\pi  & \sin m\sqrt{2}\pi \\
    -\sin m\sqrt{2}\pi  & \cos  m\sqrt{2}\pi \\
  \end{array}
\right)|q_0\rangle\\
&=\left(
  \begin{array}{cc}
    1  & 0 \\
   0   & 1 \\
  \end{array}
\right)|q_0\rangle=|q_0\rangle.
\end{align}
\end{proof}
\begin{Lm}\cite{Zhg12}\label{not-in-C(m)}
If the input $w\not\in C(m)$, then ${\cal A}(p,\varepsilon)$ rejects $w$ after Step 3.1 with a probability at least $1/(2(m-n)^2+1)$.
\end{Lm}

\begin{Lm}\cite{Zhg12}
 A coin flipping can be simulated by a 2QCFA  using two basis states
$|q_0\rangle$ and $|q_1\rangle$.
\end{Lm}

\begin{Lm}\label{acc-rej-C(m)}
Suppose that the length of  input $|w|=n$. If  $w\in C(m)$, then ${\cal A}(m,\varepsilon)$ accepts $w$ after Step 4 with the probability $\frac{\varepsilon}{2m^2(n+1)^2}$.  If $w\not\in C(m)$, then ${\cal A}(m,\varepsilon)$ accepts $w$, after Step 4, with a probability less than $\frac{\varepsilon}{2m^2(n+1)^2}$.
\end{Lm}

\begin{proof}
If $w\in C(m)$, then the  quantum state of ${\cal A}(m,\varepsilon)$ after Step 2 will be $|q_0\rangle$, according to Lemma \ref{in-C(m)}. In Step 3.2  two  random walks  start
at cell 1 and stop at cell 0 (with the left end-marker\ \
$|\hspace{-1.5mm}c$) or at the cell $n+1$ (with the right
end-marker $\$$). It is known, from the Markov chains theory, that the
probability of reaching  cell $n+1$ is $\frac{1}{n+1}$ (see
Chapter 14.2 in \cite{Fel67}). Therefore, the probability that the quantum state of ${\cal A}(m,\varepsilon)$ is $|q_0\rangle$  after Step 3.2  will be $\frac{1}{(n+1)^2}$.

The quantum state of ${\cal A}(m,\varepsilon)$ after Step 4 will therefore be
\begin{align}
|q\rangle&=U_{m,\varepsilon}|q_0\rangle=\left(
  \begin{array}{cc}
    \frac{1}{\sqrt{2m^2/\varepsilon}}  & -\frac{\sqrt{2m^2/\varepsilon-1}}{\sqrt{2m^2/\varepsilon}}\\
    \frac{\sqrt{2m^2/\varepsilon-1}}{\sqrt{2m^2/\varepsilon}}  & \frac{1}{\sqrt{2m^2/\varepsilon}}  \\
  \end{array}\right)|q_0\rangle\\
 &=\frac{1}{\sqrt{2m^2/\varepsilon}}|q_0\rangle+ \frac{\sqrt{2m^2/\varepsilon-1}}{\sqrt{2m^2/\varepsilon}}|q_1\rangle.
\end{align}
 The input is therefore accepted after the measurement in Step 4 with the probability $\frac{\varepsilon}{2m^2(n+1)^2}$.

 If $w\not\in C(m)$, then the probability that after Step 2 the quantum state of ${\cal A}(m,\varepsilon)$  is $|q_0\rangle$ is less than 1. Therefore, ${\cal A}(m,\varepsilon)$ accepts $w$, after Step 4, with a probability less than $\frac{\varepsilon}{2m^2(n+1)^2}$.
\end{proof}

In Step 4, if ${\cal A}(m,\varepsilon)$ accepts
the input string, then ${\cal A}(m,\varepsilon)$ halts.  Otherwise, the
resulting quantum state after the measurement is $|q_1\rangle$
and an application of the operator $|q_0\rangle\langle q_1|$ results in state $|q_0\rangle$ in which then
${\cal A}(m,\varepsilon)$ starts a new iteration.

Suppose that the length of  input $w$ is $n$. If $w\in C(m)$, then according to Lemma \ref{acc-rej-C(m)}, the probability of accepting the input in one iteration is
\begin{equation}
 P_a=\frac{\varepsilon}{2m^2(n+1)^2}=\frac{\varepsilon}{2m^2(m+1)^2}
\end{equation}
and the probability of rejecting the input in one iteration is
\begin{equation}
 P_r=0.
\end{equation}
If  the whole process is repeated for infinity, the accepting probability is
\begin{equation}
Pr[{\cal A}(m,\varepsilon)\  \text{accepts}\  w] =\sum_{i\geq
0}(1-P_a)^i(1-P_r)^iP_a
=\sum_{i\geq
0}(1-P_a)^iP_a
=\frac{P_a}{P_a}=1.
\end{equation}

If $w\not\in C(m)$, then according to Lemmas \ref{not-in-C(m)} and \ref{acc-rej-C(m)}, the probability of accepting the input in one iteration is
\begin{equation}
 P_a<\frac{\varepsilon}{2m^2(n+1)^2}
\end{equation}
and the probability of rejecting the input in one iteration is
\begin{equation}
 P_r>\frac{1}{2(m-n)^2+1}.
\end{equation}
If the whole process is repeated indefinitely, then the probability that ${\cal A}(p,\varepsilon)$ rejects the input $w$  is
\begin{equation}
Pr[{\cal A}(p,\varepsilon)\  \text{rejects}\  w] =\sum_{i\geq
0}(1-P_a)^i(1-P_r)^iP_r
=\frac{P_r}{P_a+P_r-P_aP_r}>\frac{P_r}{P_a+P_r}
\end{equation}
\begin{equation}
>\frac{\frac{1}{2(m-n)^2+1}}{\frac{\varepsilon}{2m^2(n+1)^2}+\frac{1}{2(m-n)^2+1}}=\frac{1}{\frac{{(2(m-n)^2+1)}\varepsilon}{2m^2(n+1)^2}+1}
>\frac{1}{\varepsilon+1}>1-\varepsilon.
\end{equation}

In the above inequality we have used the fact that $0<\frac{{2(m-n)^2+1}}{2m^2(n+1)^2}<1$ for any $m>0$ and $n> 0$.

{\bf Time analysis:} The expected running time of Steps 1 to 4 is ${\bf O}(n^2)$ time. The halting probability is in both cases ${\bf \Omega}\left(\frac{ \varepsilon}{m^2n^2}\right)$, and therefore the expected number of repetitions of the above process is, in both cases, ${\bf O}(\frac{m^2n^2}{ \varepsilon})$.
Hence the expected running time of ${\cal A}(m,\varepsilon)$ is
${\bf O}(\frac{ 1}{\varepsilon}m^2n^4)$.
\end{proof}

The matrices (quantum transition function) used in the  previous Theorem are  modifications of the ones used in \cite{Amb02}.
Similar  proof methods can be found in \cite{Amb02,Zhg12,ZhgQiu11}.

\begin{Rm}
Using the above theorem and the intersection property of languages recognized by 2QCFA \cite{Qiu08}, it is easy  to improve the result from \cite{Zhg12} related to the promise problem\footnote{See  page 102 in \cite{Zhg12}.} $A^{eq}(m)$ to a language $L^{eq}(m)=\{a^mb^m\}=L^{eq}\cap C(2m)$, where the language $L^{eq}=\{a^nb^n\,|\,n\in \mathbb{N} \}$.
 Therefore, the open problem  from \cite{Zhg12} is solved.
\end{Rm}

It is obvious that the number of states of a  1DFA to accept the language $C(m)$ is at least  $m$. Using a similar proof as of  Theorem \ref{th-EQtoDFA}, we  get:
\begin{Th}
For any integer $m$, any  polynomial expected running time 2PFA recognizing $C(m)$ with error probability $\varepsilon<\frac{1}{2}$ has at least $\sqrt[3]{(\log m)/b}$ states, where $b$ is a constant.
\end{Th}

The sizes of 1PFA and 1QFA recognizing languages $L(p)$ or $C(m)$ with an error $\varepsilon$ depend on the error $\varepsilon$ in most of the papers. For example, in \cite{Amb09}, the size of MO-1QFA accepting $L(p)$ with one-sided error $\varepsilon$ is $4\frac{\log 2p}{\varepsilon}$. If $\varepsilon<\frac{4}{p}$, the state complexity advantage of MO-1QFA  disappears. However, in our model, the sizes of 2QCFA do not depend on the error $\varepsilon$, which means that 2QCFA have state advantage for any $\varepsilon>0$.

\section{A trade-off property of 1QCFA}

Quantum resources are  expensive and hard to deal with. One  can expect to have   only   very limited number of  qubits in  current quantum system. In some cases, one cannot expect to have  enough qubits to solve a given problem (or to recognize a given language). It is therefore interesting to find out whether there are some trade-off  between needed quantum and classical resources. We prove in the following that it is so in some cases. Namely,
 we prove that there  exist  trade-offs  in case  1QCFA are used to accept the language $L(p)$.

\begin{Th}
  For any  integer $p\in {\mathbb{Z}}^+$ with prime factorization $p=p_1^{\alpha_1}p_2^{\alpha_2}\cdots p_s^{\alpha_s}$ ($s>1$),  for any partition $I_1, I_2$ of $\{1,\ldots, s\}$, and for $q_1=\prod_{i\in I_1}p_i^{\alpha_i}$ and $q_2=\prod_{i\in I_2}p_i^{\alpha_i}$, the language $L(p)$  can be recognized  with one-sided error $\varepsilon$ by a  1QCFA $A(q_1,q_2,\varepsilon)$ with ${\bf O}(\log{q_1})={\bf O}(\sum_{i\in I_1}{\alpha_i}\log p_i)$ quantum basis states and ${\bf O}(q_2)={\bf O}(\prod_{i\in I_2}p_i^{\alpha_i})$ classical states.
\end{Th}
\begin{proof}
Obviously, $\gcd{(q_1,q_2)}=1$ and therefore $L(p)=L(q_1)\cap L(q_2)$.  According to \cite{Amb09}, the language $L(q_1)$ can be recognized with one-sided error $\varepsilon$  by an MO-1QFA  with $4\frac{\log 2q_1}{\varepsilon}$ quantum basis states. It is well known that the language $L(q_2)$ can be recognized by a $q_2$-state 1DFA.
 A $q_1$-state MO-1QFA can be simulated by a 1QCFA with $q_1$ quantum basis states and 3 classical states \cite{ZhgQiu112}. A $q_2$-state 1DFA can be simulated by 1QCFA  with 1 quantum state and $q_2+1$ classical states \cite{ZhgQiu112}. Now we have to use the following lemma

 \begin{Lm}\cite{ZhgQiu112}\label{1QCFA}
Let a language $L_1$  be recognized with one-sided error $\varepsilon_1$ by a 1QCFA ${\cal A}_1$  with $q_1$ quantum basis states and $c_1$ classical states.  Let a language $L_2$  be recognized with one-sided error $\varepsilon_2$ by a 1QCFA ${\cal A}_2$ with $q_2$ quantum basis states and $c_2$ classical states. Then the language $L_1\cap L_2$  can be recognized  with one-sided error $\varepsilon_1+\varepsilon_2-\varepsilon_1\varepsilon_2$ by a 1QCFA ${\cal A}$  with $q_1q_2$ quantum basis states and $c_1c_2$ classical states.
 \end{Lm}

According to Lemma \ref{1QCFA}, the language $L(p)=L(q_1)\cap L(q_2)$ can be recognized  with one-sided error $\varepsilon$ by a 1QCFA $A(q_1,q_2,\varepsilon)$ with ${\bf O}(\log{q_1})={\bf O}(\sum_{i\in I_1}{\alpha_i}\log p_i)$ quantum basis states and ${\bf O}(q_2)={\bf O}(\prod_{i\in I_2}p_i^{\alpha_i})$ classical states.
\end{proof}

\section{Concluding remarks}

We have explored the state complexity of 1QCFA  a  family of  promise problems. For any $n\in {\mathbb{Z}}^+$, we have proved that the promise problem $A_{EQ}(n)$,  modeling the strings equality problem of length $n$ with restriction,  can be solved by an exact  1QCFA ${\cal A}(n)$ with $2n$ quantum basis states and ${\bf {O}}(n)$ classical states,  whereas the sizes of the corresponding 1DFA are $2^{{\bf \Omega}(n)}$. Afterwards, we have shown  state succinctness results of 2QCFA for two basic and extensively studied families of regular languages. Namely, $L(p)$ and $C(m)$. We have proved that 2QCFA can be remarkably more concise than their corresponding classical counterparts and we have also solved one  open problem from \cite{Zhg12}.
At last, we have proved that there  exist various trade-offs  for 1QCFA, which respect to the numbers of quantum and classical states needed for  the family of languages $L(p)$.

Some possible problems for future work are:
\begin{enumerate}
  \item It has been proved \cite{Kla00} that exact 1QCFA have not state complexity advantage over 1DFA in recognizing a language. How about exact 2QCFA vs. 2DFA?
  \item It  would be interesting to find out more trade-off properties of 1QCFA or 2QCFA.
\end{enumerate}

\section*{Acknowledgements}
We would like to thank Matej Pivoluska for stimulating discussions.
 The first and second authors  acknowledge support of
the Employment of Newly Graduated Doctors of Science for Scientific Excellence project/grant (CZ.1.07./2.3.00\linebreak[0]/30.0009)  of Czech Republic.
The third author  acknowledges support of the National
Natural Science Foundation of China (Nos. 61272058, 61073054).

%%-----------------------------
%%      your bibliography
%%-----------------------------


\begin{thebibliography}{ABCD}


\bibitem{AmbNay02} A. Ambainis, A. Nayak, A. Ta-Shma and U. Vazirani, Dense quantum coding and quantum
automata, Journal of the ACM {\bfseries 49} (4) (2002) 496--511.

\bibitem{Amb96} A. Ambainis, The complexity of probabilistic versus deterministic finite automata, Proceedings
of the International Symposium on Algorithms and Computation (ISAAC¡¯96),
LNCS 1178 (1996) 233--239.


\bibitem{Amb02} A. Ambainis and J. Watrous, Two-way finite automata with
quantum and classical states, Theoretical Computer Science {\bfseries 287}
(2002) 299--311.


\bibitem{Amb98} A. Ambainis and R. Freivalds, One-way quantum finite automata: strengths, weaknesses and generalizations, in:
Proceedings of the 39th Annual Symposium on Foundations of
Computer Science, IEEE Computer Society Press, Palo Alfo,
California, USA (1998) 332--341.

\bibitem{Amb09} A. Ambainis and N. Nahimovs, Improved constructions of
quantum automata, Theoretical Computer Science {\bfseries 410} (2009) 1916--1922.

\bibitem{AmYa11} A. Ambainis and A. Yakaryilmaz, Superiority of exact quantum automata for promise problems, Information Processing Letters  {\bfseries 112} (7) (2012) 289--291.


\bibitem {Ber05} A. Bertoni, C. Mereghetti and B. Palano, Small size quantum automata recognizing
some regular languages. Theoretical Computer Science {\bfseries 340} (2005) 394--407.

\bibitem{BMP06}  A. Bertoni, C. Mereghetti and b. Palano.  Some formal tools for analyzing quantum automata, Theoretical Computer Science {\bfseries 356} (2006)  14--25.

\bibitem{Bir93} J.C. Birget, State-complexity of finite-state devices, state compressibility and incompressibility, Math. Systems Theory {\bfseries 26}  (1993) 237--269.

\bibitem{Bro99}A. Brodsky and N. Pippenger, Characterizations of 1-way quantum finite automata, SIAM Journal on Computing  {\bfseries 31} (2002) 1456--1478.

\bibitem{Buh98} H. Buhrman, R. Cleve and A. Wigderson, Quantum vs. classical communication and
computation, Proceedings of 30th Annual ACM Symposium on Theory of Computing (1998)  63--68 .

\bibitem{Buh10} H. Buhrman, R. Cleve, S. Massar and  R. de Wolf, Nonlocality and Communication Complexity,  Rev. Mod. Phys. {\bfseries 82} (2010)  665--698

\bibitem {Chr86} M. Chrobak, Finite Automata and Unary Languages, Theoretical Computer
Science {\bfseries 47} (3) (1986) 149--158 .


\bibitem{DwS90} C. Dwork and L. Stockmeyer, A time-complexity gap for two-way probabilistic finite state automata, SIAM
J. Comput. {\bfseries 19} (1990) 1011--1023.


\bibitem {Fre82} R. Freivalds, On the growth of the number of states in result of determinization of
probabilistic finite automata, Automatic Control and Computer Sciences {\bfseries 3} (1982)  39--42.

\bibitem{Fre08} R. Freivalds, Non-constructive methods for finite probabilistic automata, International Journal of Foundations of Computer Science {\bfseries 19} (3) (2008) 565--580.

\bibitem{Fre09} R. Freivalds,  M. Ozols and L. Mancinska, Improved constructions of mixed state quantum automata, Theoretical Computer Science, {\bfseries 410} (20) (2009) 1923--1931.

\bibitem{Fel67}W. Feller, An Introduction to Probability Theory and its Applications, Vol. I, Wiley, New York
(1967).

\bibitem{Gru99} J. Gruska, Quantum Computing, McGraw-Hill,
London (1999).

\bibitem{Gru00} J. Gruska, Descriptional complexity issues in quantum computing, J. Automata, Languages Combin. {\bfseries 5} (3) (2000) 191--218.

\bibitem{GQZ13} J. Gruska, D. W. Qiu,  and S. G. Zheng,  Communication complexity of promise problems and their applications to finite automata, arXiv:1309.7739 (2013).

\bibitem{Hop79}J. E. Hopcroft and J. D. Ullman,  Introduction to Automata Theory, Languages, and
Computation, Addision-Wesley, New York (1979).

\bibitem{Kir01} G. A. Kiraz, Compressed Storage of
Sparse Finite-State Transducers, Proceedings
of CIAA 2001, Springer LNCS 2214 (2001) 109--121.

 \bibitem{Kla00} H. Klauck, On quantum and probabilistic communication: Las Vegas and one-way protocols, Proceedings of the 32th STOC (2000) 644-651.

\bibitem{KusNis97b} E. Kushilevitz, Communication Complexity,  Advances in Computers {\bfseries 44} (1997) 331--360.

\bibitem{Kon97} A. Kondacs and J. Watrous, On the power of quantum
finite state automata, in: Proceedings of the 38th IEEE Annual
Symposium on Foundations of Computer Science  (1997) 66--75.

\bibitem{Le06} F. Le Gall, Exponential separation of quantum and classical online space complexity, in: Proceedings of SPAA'06 (2006) 67--73.

\bibitem{Liu08} G. Liu, C. Martin-Vide, A. Salomaa and S. Yu,
State complexity of basic operations combined with reversal,
Information and Computation {\bfseries 206} (2008) 1178--186.



\bibitem {Mer00} C. Mereghetti and  G. Pighizzini, Two-way automata simulations and unary languages, J.
Autom. Lang. Comb. {\bfseries 5} (2000) 287--300.

\bibitem {Mer01} C. Mereghetti, B. Palano and G. Pighizzini,
 Note on the Succinctness of Deterministic, Nondeterministic, Probabilistic and Quantum Finite Automata. RAIRO-Inf.
Theor. Appl. {\bfseries 35} (2001) 477--490.


\bibitem {Mer02} C. Mereghetti and B. Palano, On the size of one-way quantum finite automata with periodic behaviors, Theoretical Informatics and Applications {\bfseries 36} (2002) 277--291.

\bibitem {Mil01} M. Milani and G. Pighizzini, Tight bounds on the simulation of unary probabilistic
automata by deterministic automata, J. Automata, Languages and
Combinatorics {\bfseries 6} (2001) 481--492.

 \bibitem{Mat12}P. Mateusa, D. W. Qiu and L. Z. Li, On the complexity of minimizing probabilistic and quantum automata, Information and Computation {\bfseries 218} (2012) 36--53.

\bibitem{Moo97} C. Moore and J. P. Crutchfield, Quantum automata and quantum grammars,  Theoretical Computer Science
 {\bfseries 237} (2000) 275--306.

\bibitem {Paz71} A. Paz, Introduction to Probabilistic Automata, Academic Press, New York (1971).


\bibitem{Qiu08} D. W. Qiu,  Some Observations on Two-Way Finite
Automata with Quantum and Classical States,  ICIC 2008, LNCS 5226 (2008)
1--8.



\bibitem{Rab59} M. O. Rabin and D. Scott, Finite automata and their decision problems,
IBM J. Research and Development, {\bfseries 3}(2) (1959) 115--125.

\bibitem{Nie00} M. A. Nielsen and I. L. Chuang, Quantum
Computation and Quantum Information, Cambridge University Press,
Cambridge (2000).

\bibitem{Qiu12} D. W. Qiu, L. Z. Li, P. Mateus and J. Gruska, Quantum finite automata, J.C. Wang, Editor,  CRC Handbook of Finite State Based Models and Applications, CRC Press (2012) 113--144.

\bibitem{Sal64} A. Salomaa, On the reducibility of events represented in automata. In Annales Academiae Scientiarum Fennicae, volume Series A of I. Mathematica (1964) 353.



\bibitem{Yu94} S. Yu, Q. Zhuang and K. Salomaa
The state complexity of some basic operations on regular languages,
Theoretical Computer Science {\bfseries 125} (1994) 315--328.

\bibitem{Yu95} S. Yu, State Complexity: Recent Results and Open Problems, Fundamenta Informaticae {\bfseries 64} (2005) 471--480.

\bibitem{Yak11} A. Yakaryilmaz and  A. C. Cem Say, Unbounded-error quantum computation
with small space bounds, Information and Computation {\bfseries 209} (2011) 873--892.

\bibitem{Yak10} A. Yakaryilmaz and  A. C. Cem Say, Succinctness of two-way probabilistic
and quantum finite automata, Discrete Mathematics and Theoretical
Computer Science {\bfseries 12} (4) (2010) 19--40.



\bibitem{ZhgQiu112} S.G. Zheng, D.W. Qiu, L.Z. Li and J. Gruska,  One-way finite automata with quantum and classical
states, In: H. Bordihn, M. Kutrib, and B. Truthe (Eds.),
 Dassow Festschrift 2012, LNCS 7300  (2012) 273--290.

\bibitem{Zhg12} S. G. Zheng, D. W. Qiu,  J. Gruska, L. Z. Li and P. Mateus, State succinctness of two-way finite automata with quantum and classical
states, Theoretical Computer Science, to  appear. Also arXiv:1202.2651 (2012).

\bibitem{ZhgQiu11} S. G. Zheng, D. W. Qiu  and  L. Z. Li, Some languages recognized by two-way finite automata with quantum and classical states, International Journal of Foundation of Computer Science, {\bfseries 23} (5) (2012) 1117--1129.

 \bibitem{Zhg13a} S. G. Zheng,  J. Gruska  and D. W. Qiu. Power of the interactive proof systems with verifiers
                modeled by semi-quantum two-way finite automata, arXiv:1304.3876 (2013).




\end{thebibliography}
\end{document}